\documentclass[a4paper,onecolumn,final]{comnet-plain}

\usepackage{amsmath}
\usepackage{amssymb,amsthm}
\usepackage{graphicx}
\usepackage{color}

\usepackage{booktabs}
\usepackage{enumitem}

\usepackage{hyperref}
\usepackage{xcolor}
\hypersetup{
    colorlinks,
    linkcolor={red!80!black},
    citecolor={blue!80!black},
    urlcolor={blue!80!black}
}

\usepackage[OT1]{fontenc} 

\usepackage{subcaption}

\newcommand{\Dt}{$\Delta t$}
\newcommand{\mDt}{\Delta t}

\binoppenalty=\maxdimen
\relpenalty=\maxdimen

\binoppenalty=\maxdimen
\relpenalty=\maxdimen

\bibpunct{[}{]}{,}{n}{,}{;} 

\title{The Temporal Event Graph}

\shorttitle{The Temporal Event Graph} 
\shortauthorlist{A. Mellor} 

\author{
\name{Andrew Mellor$^*$}
\address{Department of Applied Mathematics, School of Mathematics, \\ University of Leeds, Leeds LS2 9JT, U.K. \email{$^*$mmasm@leeds.ac.uk}}
}

\begin{document}

\maketitle

\begin{abstract}
	{
	Temporal networks are increasingly being used to model the interactions of complex systems. 
	Most studies require the temporal aggregation of edges (or events) into discrete time steps to perform analysis.
	In this article we describe a static, lossless, and unique representation of a temporal network, the \emph{temporal event graph} (TEG).
	The TEG describes the temporal network in terms of both the inter-event time and two-event temporal motif distributions.
	By considering these distributions in unison we provide a new method to characterise the behaviour of individuals and collectives in temporal networks as well as providing a natural decomposition of the network.
	We illustrate the utility of the TEG by providing examples on both synthetic and real temporal networks.
	}
	{
	temporal networks; temporal motifs; connectivity; inter-event times.
	}
\end{abstract}

\section{Introduction}
	\label{sec:introduction}

	Temporal networks have seen increased use in the study of dynamics in complex systems.
	This is due partly to the increase in available time-stamped data from sources such as Twitter \cite{kwak2010twitter} and proximity sensors \cite{sociopatterns} (among others), but also due to the recognition that the temporal patterns of complex systems have a major influence on the proliferation of processes on them \cite{scholtes2014causality}.
	The role that temporal patterns of connectivity plays on dynamics is not fully understood.
	However, whether or not temporal temporal networks slow down or speed up spreading dynamics is in fact dependent on the details of the spreading mechanisms and the choice of parameters \cite{masuda2013temporal,iribarren2009impact,karsai2011small,onaga2017concurrency}.
	What is clear however is that the inclusion of temporal information uncovers patterns not observable from the study of a static network alone \cite{mastrandrea2015contact}. 
	It is therefore vital to characterise and understand the structure of temporal networks.
	
	There are many ways to represent a temporally evolving network \cite{holme2013temporal}, the most common being a series of static networks representing the connections of the network within particular time intervals.
	There have been many attempts to incorporate temporal information into static networks \cite{moody2002importance,kempe2000connectivity}.
	This usually involves labelling edges with the times at which they appear \cite{cheng2003time}.
	Conversely, other approaches create static aggregations of the temporal network.
	Examples include the time aggregated graph \cite{holme2003network}, the line graph \cite{liljeros2003sexual}, and the transmission graph \cite{riolo2001methods}.
	In each of these examples the representation is not lossless, that is, the full temporal network can not be recovered from the representation.
	This also means that many different temporal networks may be equivalent when expressed in these representations.

	In this article we introduce the \emph{temporal event graph} (TEG), a lossless, static representation of a temporal network.
	The TEG is most similar to the transmission graph and describes the network in terms of the distribution of inter-event times (IETs) and two-event temporal motifs \cite{kovanen2011temporal}.
	A derivative of the TEG also plays a crucial role in the calculation of higher order temporal motifs \cite{kovanen2011temporal}.
	The analysis of temporal motifs has previously uncovered various behaviours of individuals when applied to a number of different temporal networks \cite{kovanen2013temporal,jurgens2012temporal,schneider2013unravelling,xie2014triadic}.

	In Section~\ref{sec:temporaleventgraph} we describe the temporal event graph and show that it uniquely defines a temporal network.
	In Section~\ref{sec:theoretical} we outline some of the theoretical properties of the TEG and in Section~\ref{sec:statistical} we state the statistical properties which describe the TEG before applying them in Section~\ref{sec:data} to characterise an online social network.
	Finally, in Section~\ref{sec:conclusions} we discuss the possible applications of the TEG and further generalisations.

\section{The Temporal Event Graph}
	\label{sec:temporaleventgraph}

	We consider temporal networks as a sequence of temporal events $E$. 
	Let $V \subset \mathbb{N}$ be a set of interacting nodes, and $T \subset \mathbb{R}^+_0$ a non-emptpy ordered set of interaction times, then the temporal network is defined as the tuple $G=G(V,E,T)$ where $E \subset V^2 \times T$. 
	An individual event $e_i=(u_i,v_i,t_i)\in E$ corresponds to an interaction of node $u_i$ with node $v_i$ at time $t_i$ (here assuming interaction is instantaneous and that each $t_i$ is distinct).
	The systems most suited to this representation are communication networks (letter and email correspondence, phone calls, social media etc.) and proximity networks (human contact networks) \cite{barrat2014measuring}. 
	To define the TEG we first need to be able to relate two events in a meaningful way, capturing the relationship of the nodes and the temporal proximity of the events.
	One such relation is that of \Dt{}-\emph{adjacency} \cite{kovanen2011temporal}.
	\begin{definition}
	\label{def:dt_adj}
	Two time-ordered events $e_i, e_j$ are said to be \Dt{}-\emph{adjacent} if they share at least one node $\left(\{u_i, v_i\} \cap \{u_j, v_j\} \neq \emptyset \right)$ and the time between the two events (inter-event time) is no greater than \Dt{}, i.e., $0 < t_j-t_i < \mDt$.
	\end{definition}
	\noindent With this definition we can formally define the TEG.
	\begin{definition}
	\label{def:dt_teg}
	For a temporal network $G = G(V, E, T)$, the \Dt-\emph{Temporal Event Graph}, hereby known as the \Dt-TEG, is a directed graph $\mathcal{G} = \mathcal{G}(\mathcal{V}, \mathcal{E})$ with $\mathcal{V} = E$ and $\mathcal{E} \subset \mathcal{V} \times \mathcal{V}$.
	The graph is defined such that there is a vertex for each event in $E$ and each vertex is connected to the \emph{subsequent} \Dt-adjacent event of each node in that event. 
	More precisely, let
	\begin{align*}
	S(u_i) &= \{ k | \phantom{1}\underbrace{( \{u_i\} \cap \{u_k, v_k\} \neq \emptyset )}_{\textrm{Share a node}}\phantom{.} \text{ and } \underbrace{(0 < t_k - t_i < \mDt)}_{\textrm{Occur within }\mDt \textrm{ of each other}}\}, 
	\end{align*}
	be the set of subsequent \Dt{}-adjacent events for the node $u_i$ with the equivalent set defined for $v_i$.
	Then the set of edges in the TEG is then given by
	\begin{align*}
	\mathcal{E} &= \{ (e_i, e_j) | (j = \min\{S(u_i)\}) \text{ or } (j = \min\{S(v_i)\}) \}. 
	\end{align*}
	\noindent This construction means that each vertex has an out-degree and in-degree of at most two (see Lemma~\ref{lem:degree}). 
	\end{definition}

	The \Dt-TEG consists of one or more connected temporal components (or maximal temporal subgraphs \cite{kovanen2011temporal}), that is, for each pair of events in a component there exists a sequence of events between them such that all pairs of consecutive events are \Dt-adjacent, that is each pair of events are \Dt-\emph{connected}.
	Of particular interest is the \Dt-TEG in the limit $\mDt \to \infty$, hereby referred to as the TEG.
	The examples in Figure~\ref{fig:duality} show how the TEG is constructed from an event sequence.
	To avoid ambiguity we use the terms nodes and events for the temporal network, and vertices and edges for the TEG.
	\begin{figure}[p]
	\centering
	\includegraphics[width=\linewidth]{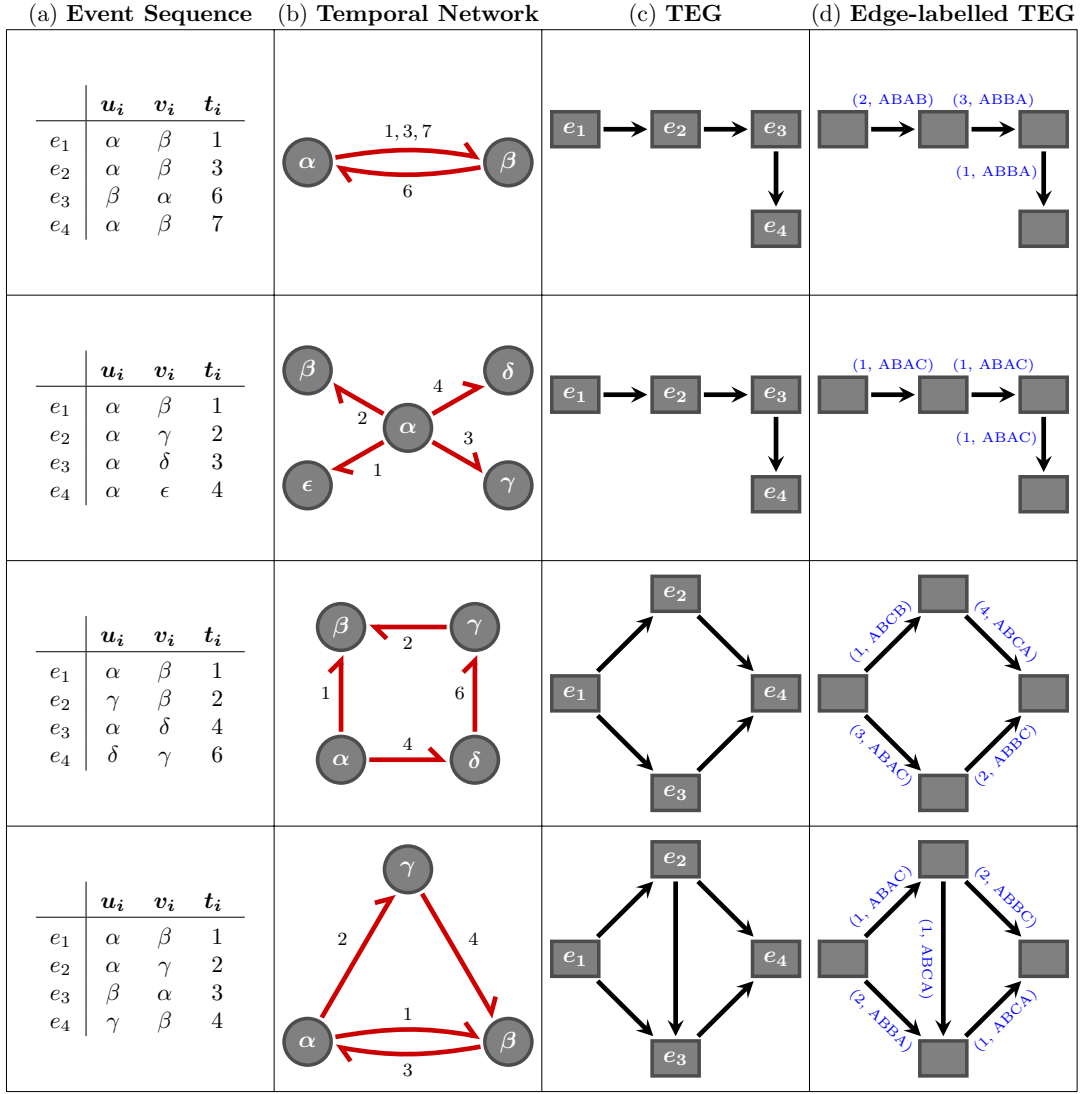}
	\caption{
	Illustration of the duality of temporal networks and the temporal event graph.
	a) Four simple temporal networks (event sequences) involving four events. 
	b) Pictorial representations of the temporal networks. Event labels represent the instantaneous time when that event occurred between two nodes. 
	c) The TEG for each temporal network.
	d) The corresponding \emph{edge-labelled} TEGs (Def.~\ref{def:edge-labelled_teg}).
	Edges are labelled with the tuple $(\tau, \mu)$, the inter-event time and motif respectively.
	Note in the bottom example the next two events for node A are connected to the first event. 
	This is consistent as the ABBA edge occurred \emph{after} that of the ABAC, i.e., node A's subsequent event was $\textrm{A} \to \textrm{C}$ and node B's subsequent event was $\textrm{B} \to \textrm{A}$ (coincidently A's next event). 
	}
	\label{fig:duality}
	\end{figure}
	
	There are two important functions of the edge set to consider.
	Firstly the function $\tau: \mathcal{E} \to \mathbb{R}_{0}^+$, given by $\tau\left((e_i, e_j)\right) = t_j - t_i$ describes the IET between the two events.
	The function $\mu: \mathcal{E} \to \mathbb{M}$, where $\mathbb{M} = \{\textrm{ABAB, ABBA, ABAC, ABCA, ABBC, ABCB}\}$ is the set of two-event motifs (Table \ref{table:twoeventmotifs}), describes the relative positions of the nodes between events.
	These motifs are given a descriptive name which is indicative of the behaviour associated with each pattern.
	For example, the ABAC motif is described as the broadcast motif as node A is in contact with multiple other nodes (B and C).
	The ABBA motif is the reciprocal motif, as an event from A to B is then followed by the reciprocal event B to A. 
	Let $f_{ij}$ be an enumeration of the ordered sequence of nodes $(u_i, v_i, u_j, v_j)$ (not necessarily distinct) mapped to the corresponding alphabetic character\footnote{
		e.g. $f_{ij}(u_i) = \text{A}, f_{ij}(v_i) = \text{B}, \dots$.
	}, then $\mu\left((e_i, e_j)\right) = f_{ij}(u_i)f_{ij}(v_i)f_{ij}(u_j)f_{ij}(v_j)$.
	For example, the edge $((5,10,t_0), (10,12,t_1))$ has nodes ($5,10,12$) which are mapped to (A,B,C) and so becomes ABBC under the action of $\mu$.
	It is also possible for the motif function $\mu$ to incorporate other event data such as event or node colourings.

	\begin{table}[h!]
		\centering
		\caption{
		The set of all possible two-event motifs $\mathbb{M}$, given by their contact sequence, description, label, and label properties $\xi_{\rm in}$, $\xi_{\rm out}$, and $\xi_{\rm switch}$. 
		}
		\label{table:twoeventmotifs}
		\begin{tabular}{@{}lllrrr@{}} \toprule[1pt]
		\textbf{Motif} & \textbf{Name}& \textbf{Shorthand} & $\xi_{\rm out}$ & $\xi_{\rm in}$ & $\xi_{\rm switch}$ \\ \midrule[0.5pt]
		$A\to B$, $B\to A$ \hspace{1ex} & Reciprocal & ABBA & AB & BA & $-1$ \\
		$A\to B$, $A\to B$ & Repeated & ABAB & AB & AB & $1$\\
		$A\to B$, $A\to C$ & Broadcasting & ABAC & A & A & $1$\\
		$A\to B$, $C\to A$ & Non-sequential & ABCA & A & B & $-1$\\
		$A\to B$, $B\to C$ & Message Passing \hspace{1ex} & ABBC & B & A & $-1$\\
		$A\to B$, $C\to B$ & Receiving & ABCB & B & B & $1$\\ \bottomrule[0.5pt]
		\end{tabular}
	\end{table}

	There are three properties of the motif set, $(\xi_{\rm out}, \xi_{\rm in}, \xi_{\rm switch})$, which are required in Section~\ref{subsec:duality} to describe consistency conditions for the TEG.
	For event pairs involving three distinct nodes we define $\xi_{\rm out}$ to be the label of the node which appears in both events, $\xi_{\rm in}$ to be the position of the shared node in the later event, and $\xi_{\rm switch} = 1$ if $\xi_{\rm out}=\xi_{\rm in}$ and $-1$ otherwise.
	For example in the motif ABBC the node labelled B occurs in both events so $\xi_{\rm out}(\textrm{ABBC})=\textrm{B}$ and takes the first position in the second event so $\xi_{\rm in}(\textrm{ABBC})= \textrm{A}$.
	Subsequently as $\xi_{\rm out}\neq\xi_{\rm in}$, then $\xi_{\rm switch}(\textrm{ABBC})=-1$, the node labelled B has switched between being a target to a source.
	For consistency we define $\xi_{\rm out}(\textrm{ABAB}) = \textrm{AB} = \xi_{\rm out}(\textrm{ABBA})$ and $\xi_{\rm in}(\textrm{ABAB})=\textrm{AB}$ and $\xi_{\rm in}(\textrm{ABBA}) = \textrm{BA}$.
	
	\subsection{Duality}
	\label{subsec:duality}

	When constructing the TEG from a temporal network we have information on the events and their connectivity.
	We can also consider a TEG without the event information, defined purely by the connectivity information and edge functions (IETs and motifs).
	\begin{definition}
	\label{def:edge-labelled_teg}
	The \emph{edge-labelled TEG} is the static graph defined by the upper-triangular adjacency pair $(A^{\tau}, A^{\mu})$
	where
	\begin{align*}
	A^{\tau}_{ij} = \begin{cases} 
			  \tau(e_i, e_j) &\mbox{if } (e_i, e_j) \in \mathcal{E} \\ 
			  0 & \mbox{otherwise } 
			  \end{cases},
	\end{align*}
	is the weighted adjacency matrix consisting of IETs and
	\begin{align*}
	A^{\mu}_{ij} = \begin{cases} 
			  \mu(e_i, e_j) &\mbox{if } (e_i, e_j) \in \mathcal{E} \\ 
			  0 & \mbox{otherwise } 
			  \end{cases}.
	\end{align*}
	is the matrix containing edge motif labels.
	\end{definition}

	Not all permutations of the vertices and edges of an edge-labelled TEG describe a temporal network.
	There are four conditions required for an edge-labelled TEG to represent a temporal network. 
	
	\begin{enumerate}
	\item[{[C1]}] \textbf{Event times must be consistent across all paths:} Let $P_{ij}$ be the set of all directed paths between vertices $i$ and $j$. 
	We describe a path $p_\alpha \in P_{ij}$ as the sequence of edges in the path. 
	The sum of inter-event times along all paths must be equal, that is 
	\begin{align*}
	\sum_{(k,l) \in p_\alpha}{A^\tau_{kl}} = \sum_{(k,l) \in p_\beta}{A^\tau_{kl}} \textrm{ for all } p_\alpha, p_\beta \in P_{ij}.
	\end{align*}
	\item[{[C2]}] \textbf{Nodes in each event have only one subsequent event:}
	For each pair of out-edges $(i, k), (i, l)$ of a vertex we require $\xi_{\rm out}(A^\mu_{ik}) \neq \xi_{\rm out}(A^\mu_{il})$.

	\item[{[C3]}] \textbf{Nodes in each event cannot be overprescribed:}
	For each pair of in-edges $(k, i), (l, i)$ of a vertex we require $\xi_{\rm in}(A^\mu_{ki}) \neq \xi_{\rm in}(A^\mu_{li})$. 

	\item[{[C4]}] \textbf{Edge types and nodes must be consistent across multiple paths:} If there exists an edge $(i,j)$ such that there exists a secondary path $p \in P_{ij}$ via at least one other vertex then
	\begin{align*}
	A^\mu_{ij} = \begin{cases} 
			  \textrm{ABAB} &\mbox{if } \displaystyle\prod_{(k,l) \in p}{\xi_{\rm switch}(A^\mu_{kl})} = 1 \\ 
			  \textrm{ABBA} & \mbox{if } \displaystyle\prod_{(k,l) \in p}{\xi_{\rm switch}(A^\mu_{kl})} = -1
			  \end{cases} .
	\end{align*}
	Conversely if there is a vertex with two in edges, $(i,j), (k,j)$, with $A^\mu_{ij} \in \{\textrm{ABAB, ABBA}\}$ then there exists a path $p \in P_{ij}$ with $(k,j) \in p$ and $\prod_{(m,n) \in p}{\xi_{\textrm{switch}}(A^\mu_{mn})} = \xi_{\textrm{switch}}(A^\mu_{ij})$.
	Similarly for a vertex with two out edges $(i,j), (i,k)$ with $A^\mu_{ij} \in \{\textrm{ABAB, ABBA}\}$ then there exists a path $p \in P_{ij}$ with $(i,k) \in p$ and $\prod_{(m,n) \in p}{\xi_{\textrm{switch}}(A^\mu_{mn})} = \xi_{\textrm{switch}}(A^\mu_{ij})$.
	\end{enumerate}
	We call graphs which satisfy these conditions \emph{consistent} graphs.
	Those graphs which do not satisfy these conditions are inconsistent in that they do not uniquely describe a temporal network, and attempting to recover the temporal network will lead to contradiction.
	Examples of inconsistent TEGs are given in Figure~\ref{fig:inconsistent_graphs}.
	Generally it is difficult to generate graphs which satisfy these conditions however any TEG generated from a temporal network will be consistent.

	\begin{figure}[h!]
	\centering
	\includegraphics[width=0.7\linewidth]{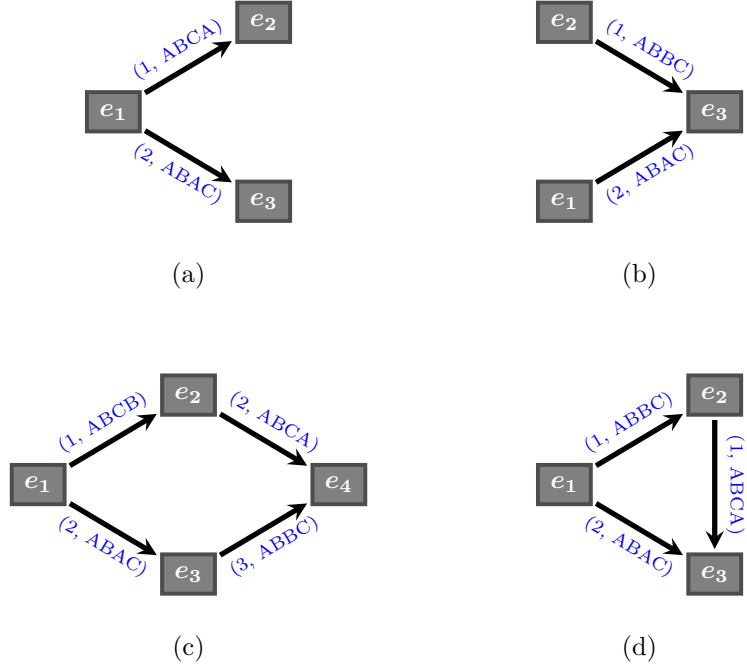}
	\caption{
	Inconsistent edge-labelled temporal event graphs. Edges are labelled with the tuple $(\tau, \mu)$. 
	a) The subsequent two events for node A are included as edges, breaking condition [C2]. 
	b) Both incoming edge types dictate the first node of the event which is contradictory (condition [C3]). 
	c) The inter-event times across multiple paths are not equal (condition [C1]). 
	d) The edge between events $e_1$ and $e_3$ is incorrectly labelled. 
	By reconstructing the temporal network or using condition [C4] we see that $A^\mu_{13}= \textrm{ABAB}$.
	}
	\label{fig:inconsistent_graphs}
	\end{figure}

	For each connected component of an edge-labelled TEG we are able to reconstruct the temporal network with the following algorithm:
	\begin{enumerate}
	\item[(a)] Find the maximal path from a root vertex (no incoming edges) to a leaf vertex (no outgoing edges) in the edge-labelled TEG using the network of IETs, $A^\tau$, allowing for backwards traversal along edges with opposite weight. (Fig.~\ref{fig:algorithm}(a)). This can be achieved by finding the shortest path in the network $\left(A^\tau\right)^T - A^\tau$. 
	\item[(b)] Label the first vertex in the path with $t=0$ and subsequently propagate the event times through the edge-labelled TEG along the edges.
	For a vertex $i$, the time at which that event occurs is given by
	\begin{align*}
	t_i = \sum_{(m,n) \in P_{0i}} \left(\left(A^\tau\right)^T - A^\tau\right)_{mn}
	\end{align*}
	To be able to do this we require the condition [C1] otherwise the existence of multiple paths between vertices leads to a contradiction in event times. 
	\item[(c)] For events in time order, resolve the nodes in each event from the incoming edges (Fig.~\ref{fig:algorithm}(b,c)).
	We require condition [C3] here otherwise there can be a conflict on resolving a node position.
	If a node in an event is unprescribed (the event has zero or one incoming edge) then the unprescribed nodes are given a new label.
	\end{enumerate}
	Condition [C2] is required by definition of the edge-labelled TEG to enforce that the subsequent event of each node is connected by an edge. 
	Without it, the subsequent two edges for one node could be given.
	Finally condition [C4] ensures that the edge-labelled TEG is uniquely labelled (Fig.~\ref{fig:inconsistent_graphs}(d)).

	\begin{figure}[t!]
	\centering
	\includegraphics[width=\linewidth]{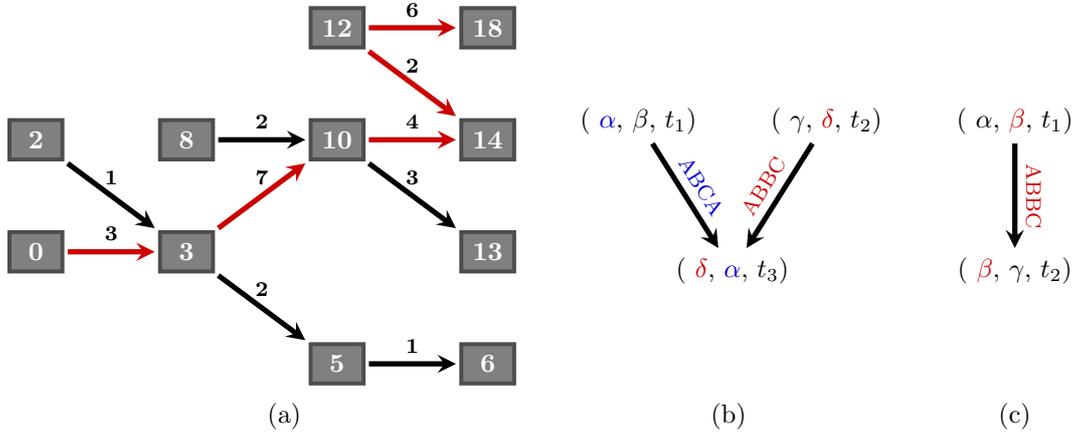}
	\caption{
	The inverse algorithm for the TEG. 
	a) The maximal path between root and leaf vertices (blue) through the TEG with edges labelled with IETs.
	Once the maximal path has been found, the root vertex is assigned time $t=0$ and the remainder of times are found by propagation along edges (red).
	b) The resolution of an event from two incoming edges. Each incoming edge determines one of the nodes in the later event.
	c) The resolution of an event with one incoming edge. In this case only one node is prescribed and so the other is given a new label.
	}
	\label{fig:algorithm}
	\end{figure}

	The existence of an inverse algorithm confirms a duality between the edge-labelled TEG and the temporal network.
	Step (a) of the inverse algorithm assumes that the maximal path through the TEG contains the earliest and latest event in the temporal network, which we now prove.
 	\begin{lemma}
	\label{lemma:maximal_path}
	The maximal path (allowing for backwards traversal along edges with negative weight) through the edge-labelled TEG includes the earliest and latest event in the temporal network.
	\end{lemma}
	\begin{proof}
	Let $p_{max} = (e_0, \dots, e_k)$ be the sequence of vertices in the maximal path.
	Suppose there exists an event $e_{*} \notin p_{max}$ such that $t_* < t_i$ for $i=0,\dots,k$.
	Then, as the TEG is connected, there exists a path $p_{*i}$ (ignoring edge directions) from $e_* \to e_i$, $\forall e_i \in p_{max}$.
	Then $l(p_{*i}) > l(p_{0i})$ where $l(\cdot)$ is the weighted length of the path, and hence the path $e_* \to e_i \to e_k$ is longer than $p_{max}$.
	This is a contradiction and hence the maximal path through the TEG must contain the earliest event in the temporal network.
	A similar but opposite argument shows that the latest event is also contained in the maximal path. 
	\end{proof}	

	\begin{theorem}
	\label{thm:uniqueness}
	Let $X$ be the set of all temporal networks translated in time such that the first event occurs at $t=0$, nodes are labelled in order of appearance, and such that the time-aggregated graph of connections is connected.
	Let $Y$ be the set of all consistent weakly connected temporal event graphs.
	Then there exists a bijection $f: X \to Y$, that is, an edge-labelled TEG uniquely describes a temporal network in $X$.
	\end{theorem}
	\begin{proof}
	Trivially for each temporal network there exists only one edge-labelled TEG as the nodes in each event have at most one subsequent event\footnote{
	Here we assume that no two events occur at the time same.
	} and the functions $\tau$ and $\mu$ are deterministic.
	The proof rests on the existence of the inverse algorithm $f^{-1}$, outlined above. 
	We consider a general event $e_i = (u_i, v_i, t_i)$ in the temporal network, and its representative vertex $x$ in the edge-labelled TEG.
	By the translation of the temporal network, this event occurs $t_i$ time units after the first event.
	By finding the maximal path through the edge-labelled TEG we find the first event in the temporal network (Lemma \ref{lemma:maximal_path}), and can hence find the time which $x$ occurs relative to this first event, that is, $t_i$.
	The event is now is correctly placed in time.
	To recover the nodes of the event $u_i$ and $v_i$, assume the nodes in all previous events have been correctly determined in order of appearance. 
	There are three possible cases:
	\begin{enumerate}
	\item Event $e_i$ has no incoming edges. In this case neither of these nodes have previously interacted and can be enumerated.
	\item Event $e_i$ has one incoming edge prescribing one node. In this case a new node is involved and is enumerated accordingly (Fig.~\ref{fig:algorithm}(c)). 
	\item Event $e_i$ has one or two incoming edges prescribing both nodes. In this case the nodes are completely determined by previous events (Fig.~\ref{fig:algorithm}(b)).
	\end{enumerate}
	For the base case, the earliest event vertices have no incoming edges and are labelled freely.
	Subsequent event vertices must then have all incoming edges prescribed as they occur strictly earlier in time.
	Hence the nodes in $e_i$ are correctly labelled, relative to the labelling of the previous events.
	As both nodes are labelled relative to previous events, and the time of the event is positioned relative to the first event, the event is recovered from the TEG.
	Since this argument holds for an arbitrary event in the temporal network, it holds for all.
	Therefore $f^{-1}(f(X)) = X$, and $f$ is a bijection. 
	\end{proof}
	\begin{corollary}
	A temporal event graph $\mathcal{G}$, consisting of multiple connected components defines a temporal network up to a translation of time between components. 
	If the events of $\mathcal{G}$ are time stamped then $\mathcal{G}$ uniquely defines a temporal network.
	\end{corollary}
	\begin{proof}
	By Theorem~\ref{thm:uniqueness} for each connected component there exists a unique temporal network such that the earliest event occurs at $t=0$. 
	Trivially there exists an ensemble of temporal networks with the same TEG, dependent on the choice of earliest event time for each component.
	If the time of this event is given then the choice is removed and hence the TEG uniquely defines the temporal network.
	\end{proof}

	Time translation between components may seem disconcerting, however these components are truly disconnected and do not share \emph{any} nodes.
	This means that, assuming the network is not visible to those within it, any dynamics on the network are completely independent across components\footnote{
	In the case where the network is visible, observing the network usually prompts a response that is directed towards the observed agents, subsequently connecting the two components.
	There may be cases where nodes in one component observe nodes in another and act upon that information without any interaction with the component.
	In these cases it is important to include the time stamp of each event in the TEG. 
	}.
	Most digital communication channels that we will consider are hidden from an observer, e.g. email, SMS, telephone calls. 
	Other sources of communication such as Twitter are in the public domain and so all messages are observable (although require active searching).
	Furthermore, with real examples we keep the event timestamps which fixes the temporal components in time, and so the TEG uniquely defines a temporal network.
	This means that the temporal network can be uniquely defined within the time translation of components by the network of subsequent adjacent events, their IETs, and the motifs formed between them.
	As a result, considering the network in this formalism is equivalent to studying the temporal network as the same information is contained in both.

\section{Theoretical Properties of the TEG}
	\label{sec:theoretical}

	In this section we state and prove a number of properties of the TEG.
	\begin{lemma}
	Each vertex in the TEG has at most in-degree two and out-degree two.
	\label{lem:degree}
	\end{lemma}
	\begin{proof}
	Consider an event vertex representing the event $e_i = (u_i, v_i, t_i)$.
	From our definition we let
	\begin{align*}
	A^+_{u}(i) &= \{ k | ( \{u_i\} \cap \{u_k, v_k\} \neq \emptyset ) \text{ and } (0 < t_k - t_i < \mDt) \}, \\
	A^+_{v}(i) &= \{ k | ( \{v_i\} \cap \{u_k, v_k\} \neq \emptyset ) \text{ and } (0 < t_k - t_i < \mDt) \} 
	\end{align*}
	be the subsequent \Dt{}-adjacent events for the nodes $u_i$ and $v_i$ respectively.
	The set of edges in the TEG is given by
	\begin{align*}
	\mathcal{E} &= \{ (e_i, e_j) | j = \min(A^+_{u}(i)) \text{ or } j = \min(A^+_{v}(i)) \}. 
	\end{align*}
	Therefore, for each $e_i$ there exists then the two edges to events whose indices are the minima of each set.
	These two minima do not need be unique, nor exist, and so there are at most two out edges.
	For the edge in-degree, the previous \Dt{}-adjacent events for the nodes $u_i$ and $v_i$
	are
	\begin{align*}
	A^-_{u}(i) &= \{ k | ( \{u_i\} \cap \{u_k, v_k\} \neq \emptyset ) \text{ and } (0 < t_i - t_k < \mDt) \}, \\
	A^-_{v}(i) &= \{ k | ( \{v_i\} \cap \{u_k, v_k\} \neq \emptyset ) \text{ and } (0 < t_i - t_k < \mDt) \} 
	\end{align*}
	We can analogously define the edge set as
	\begin{align*}
	\mathcal{E} &= \{ (e_j, e_i) | j = \max(A^-_{u}(i)) \text{ or } j = \max(A^-_{v}(i)) \}. 
	\end{align*}
	By the same reasoning as with forward definition, vertices can have a maximum in-degree
	of at most two.
	\end{proof}

	\begin{lemma}
	The TEG is a directed acyclic graph (DAG).
	\label{lem:dag}
	\end{lemma}
	\begin{proof}
	For a graph $G$ to be a DAG, each node in $G$ must not have a directed path from that node back to itself.
	The edge set is given by
	\begin{align*}
	\mathcal{E} &= \{ (e_i, e_j) | j = \min(A^+_{u}(i)) \textrm{ or } j = \min(A^+_{v}(i)) \}
	\end{align*}
	where the set $A^+_{u}(i)$ contains only events $e_k$ such that $t_k > t_i$ by definition.
	Suppose there exists a direct path from event $i$ back to itself via a sequence of ordered events $e_{k_1}, e_{k_2}, \dots e_{k_n}$. 
	Then by transitivity this implies $t_i < t_{k_1} < t_{k_2} < \dots < t_{k_n} < t_i$, which is a contradiction.
	Hence no such path exists and the TEG is a DAG.
	Simply put, as edges travel strictly forward in time there can be no cycles in the graph.
	\end{proof}

	\begin{lemma}
	The set of nodes in each component of the TEG are distinct, i.e. if there exists two components of the TEG, $C_1$ and $C_2$, with node sets $V_1, V_2, \subset V$ then $V_1 \cap V_2 = \emptyset$.
	\label{lem:distinct}
	\end{lemma}
	\begin{proof}
	Suppose $V_1 \cap V_2 \neq \emptyset$ and there exists a node $u \in V_1 \cap V_2$.
	Then there exists a set of events in $C_1$ which contain $u$ with times $t^{(1)}_1, t^{(1)}_2, \dots, t^{(1)}_{n_1}$.
	Similarly there exists a set of events in $C_2$ which contain $u$ with times $t^{(2)}_1, t^{(2)}_2, \dots, t^{(2)}_{n_2}$.
	Assuming that event times are distinct then there exists an ordering of these times.
	Regardless of the relative ordering of these times there must exist a time $t^{(1)}_i$ followed by a time $t^{(2)}_j$ (or vice versa).
	These events share a node and the timing of the events are consecutive meaning the two events are \emph{adjacent}.
	This implies there exists an edge between the two events by definition of the TEG, and $C_1$ and $C_2$ are one component.
	This contradicts the original statement and hence $C_1$ and $C_2$ must contain distinct nodes. 
	\end{proof}

	\noindent Note that this is not true in the \Dt-TEG, even if the components completely overlap in time.

\section{Statistical Properties of the TEG}
	\label{sec:statistical}

	The \Dt-TEG provides a means of assessing the temporal structure of the network.
	In this section we consider the \Dt-TEG as a weighted directed static network where edge weights are the IETs between events.
	This allows us to prune a network based on edge weights (IETs).
	We consider the \emph{weakly connected components} of the TEG where two vertices are in the same component if they are connected on the undirected graph induced by ignoring edge direction.
	Note that these components are equivalent to the maximal \Dt-connected subgraphs defined in \cite{kovanen2011temporal} and describe the connectivity of the events themselves however they cannot describe the connectivity of the nodes in general\footnote{
	One can make a number of statements on the connectivity of the nodes, following edges with certain motif types. 
	For example, a chain of ABBC edges implies a path from the source node of the first event to the target node of the final event.
	However this is difficult to do in fully generality.
	}.
	In fact, finding strongly connected components of nodes in temporal networks has been shown to be an NP-complete problem \cite{nicosia2012components}.

	The \Dt-TEG contains edges $(i,j)$ where $A^\tau_{ij} < \mDt$, using the notation of the edge-labelled TEG from the previous section.
	Let $C_i^{\mDt}$ be the $i$th component of the \Dt-TEG, where components are partially ordered by the number of events they contain such that $|C_0| \ge |C_1| \ge \dots$.
	These components provide a natural decomposition of the temporal network.

	\subsection{Component Sizes, Distribution, and Growth}
	\label{subsec:component_sizes}

	The number and size of components in the \Dt-TEG is dependent on \Dt.
	A natural question is to ask how many components there are in a temporal network and how the events are distributed between them.

	In the limit $\mDt \to 0$ the TEG will be completely disconnected (assuming no two events occur at once), however it is not guaranteed that as $\mDt \to \infty$ that a single component will form.
	In fact in the limit $\mDt \to \infty$ the components of the TEG contain distinct sets of nodes (Lemma~\ref{lem:distinct}) and correspond to the connected components of the time-aggregated temporal network.
	For intermediate \Dt{} the structure of the TEG has a complex dependency on both the connectivity of the nodes (network topology) and the times between subsequent events.

	To characterise the network structure we look at the component size distribution of the \Dt{}-TEG.
	We are also interested in the size of the largest component $|C_0^{\mDt}|$, given by the number of events in that component. 
	In particular, understanding the growth of $|C_0^{\mDt}|$ as a function of \Dt{} gives clues to the temporal structure of the network, e.g. what fraction of the whole network does it contain and for what value of \Dt{} does it reach $95\%$ of its total size?

	As an example, we look at a randomly generated temporal network.
	To generate a temporal network of $N$ nodes with $M$ events with a prescribed IET distribution $X$ we perform the iteration:
	\begin{enumerate}
	\item Increment $t$ to $t+\tau$ where $\tau \sim X$
	\item Draw $u, v$ from $\{1, \dots, N \}$ without replacement
	\item Add event $(u,v,t)$ to the temporal network
	\end{enumerate}
	for each event, after initialising $t=0$.

	\begin{figure}[htb]
	\centering
	\begin{subfigure}{0.49\linewidth}
	\centering
	\includegraphics[width=\linewidth]{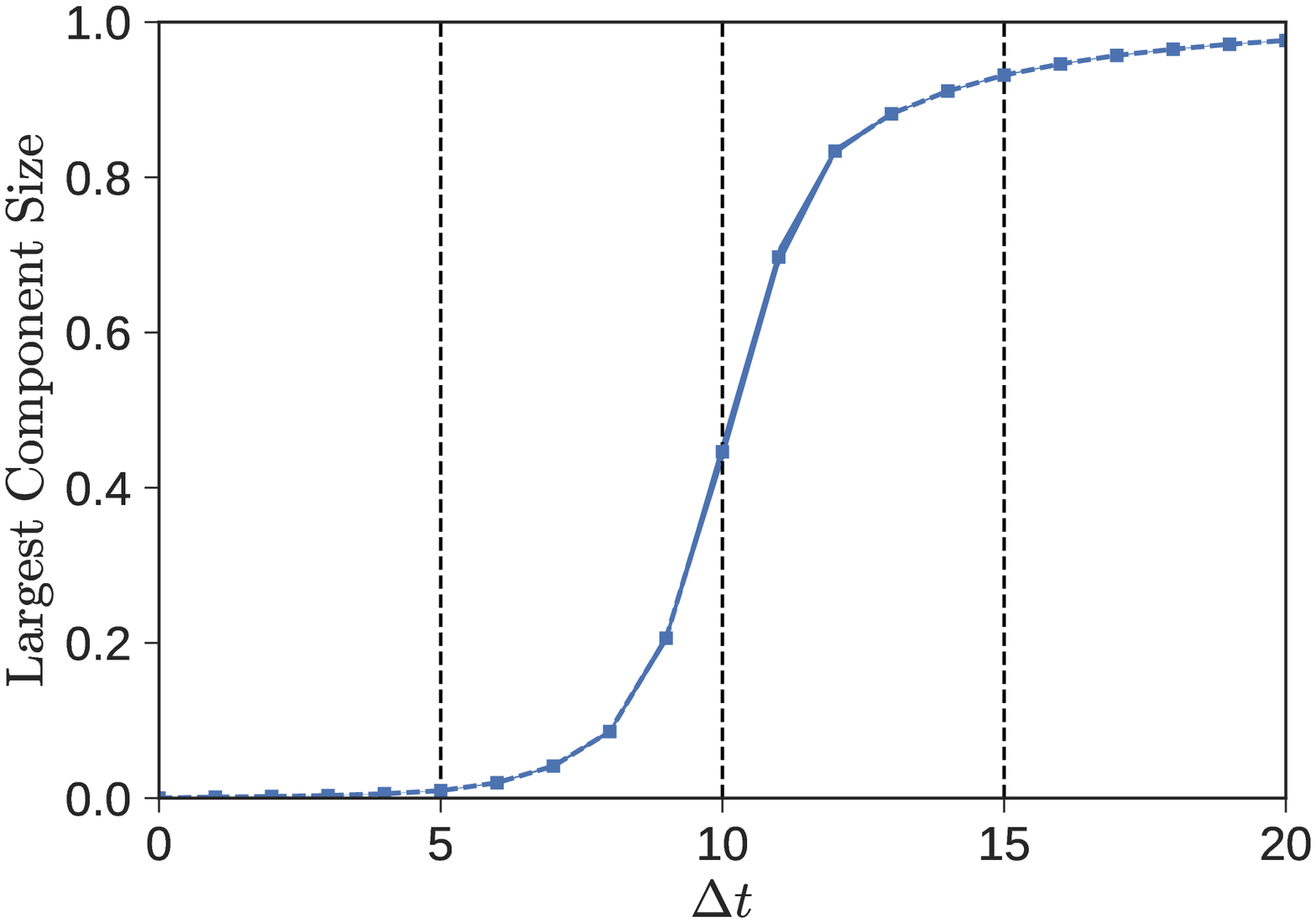}
	\caption{}
	\label{fig:random_components_a}
	\end{subfigure}
	\hfill
	\begin{subfigure}{0.49\linewidth}
	\centering
	\includegraphics[width=\linewidth]{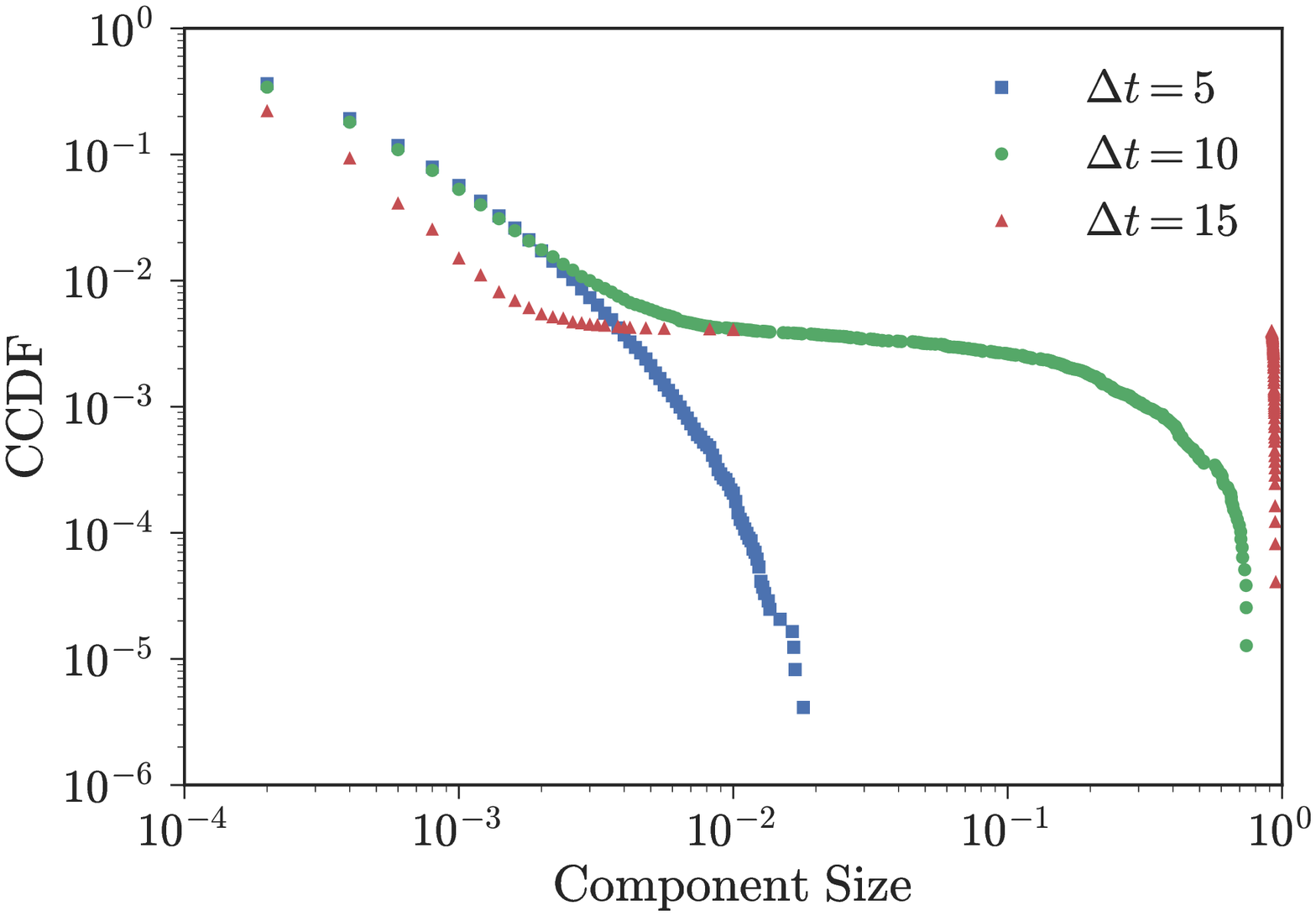}
	\caption{}
	\label{fig:random_components_b}
	\end{subfigure}
	\caption{Temporal component dependence on $\Delta t$. 
	a) The size of the largest temporal component in the \Dt{}-TEG as a fraction of all events for a random temporal network of $200$ nodes and $5000$ events.
	The largest component size has a sigmoidal dependence on \Dt{}, with a sharp transitional period from being only a small fraction of all events ($< 10\%$), to containing almost all events ($> 90\%$).
	b) The corresponding distribution of temporal component sizes for $\mDt = 5,10,15$ constructed using an ensemble of random temporal networks.
	For $\mDt=5$ there are a range of component sizes however non which make up more than $10\%$ of the network.
	For $\mDt=10$ components can take any size.
	For $\mDt=15$ components either make up the majority of the network, or are small isolated components.
	}
	\label{fig:random_components}
	\end{figure}

	In Figure~\ref{fig:random_components} we see the results for a random graph where $N=200$, $M=5000$, and $X$ is power-law distributed with density $P(x;a) = ax^{a-1}$, where $0 \le x \le 1$ and $a=0.2$.
	Results are averaged over an ensemble of $100$ temporal networks.
	The size of the largest component has a sigmoidal dependence on \Dt{}, with only a small fraction of the TEG connected below a characteristic time, and the majority of events connected above (Fig.~\ref{fig:random_components}(a)).
	The average duration of the temporal network is $1000$ meaning that when \Dt{} is only $2\%$ of the network duration, the majority of events are connected.
	Also, due to the random selection of nodes the largest component ultimately contains every event as $\mDt \to \infty$.
	The distribution of temporal components (Fig.~\ref{fig:random_components}(b)) also display this transition.
	For $\mDt = 5$ there is a continuous spectrum of component sizes although the maximum observed size is less than $10\%$ of events.
	The probability of observing components any larger grows exponentially small.
	For $\mDt = 10$ almost all possible component sizes are observed.
	However, above the characteristic time at $\mDt = 15$, the distribution is not continuous.
	Components either are a small fraction of the TEG, or are the majority fraction. 
	There are no components of intermediate size.

	Another aspect of a component size is its growth over time.
	As events are added to the temporal network, they may be added to one of the existing temporal components if they are \Dt{}-connected to an event in those components.
	There may however be events introduced which are \Dt{}-connected to two existing temporal components.
	These events cause the coalescence of the two components.
	By studying the growth of components over time we can observe the events which bring different parts of the network together and understand how the network grows over time.

	One way to visualise the temporal components is through a \emph{temporal barcode}, as seen in Figure~\ref{fig:barcode}.
	This displays the components of the \Dt-TEG, ordered by their size with the largest components at the bottom.
	Within each component, the individual events are plotted by a single vertical line.
	This visualisation allows us to see the duration of each component, its temporal position relative to other components, and the distribution of IETs within the component.
	
	\begin{figure}[h!]
	\centering
	\includegraphics[width=0.9\linewidth]{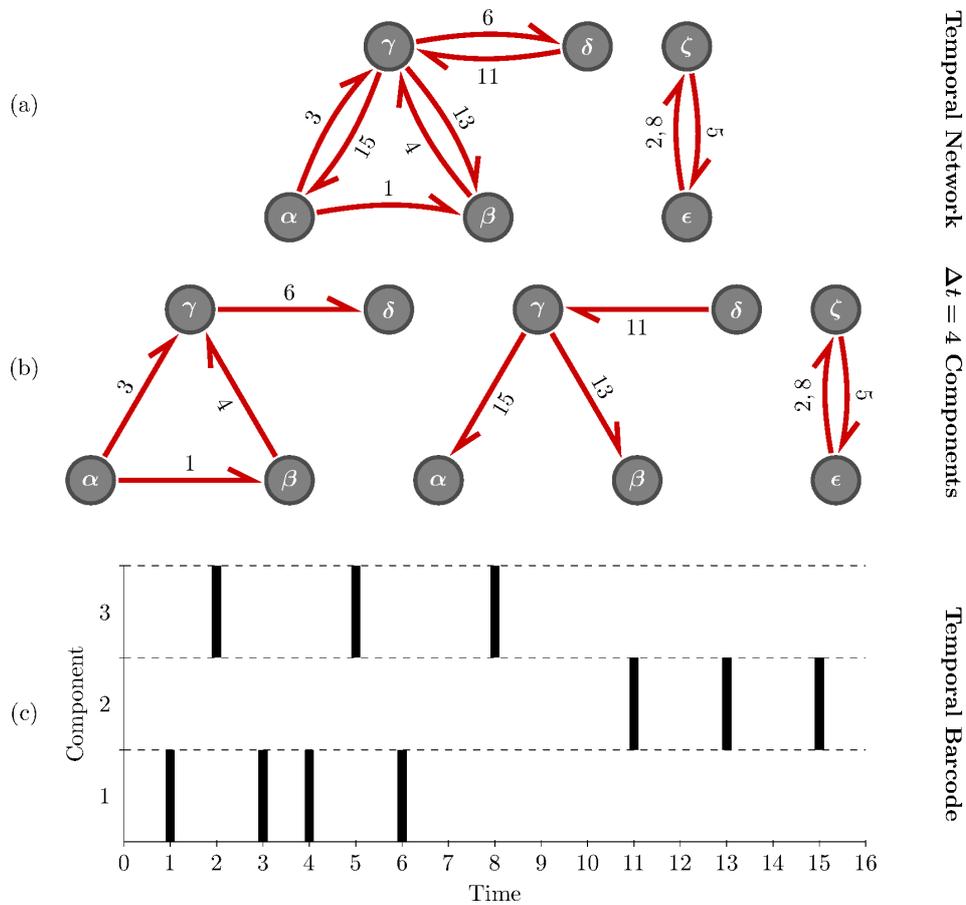}
	\caption{
	Illustration of the temporal barcode associated with a \Dt-TEG. 
	a) A temporal network involving six nodes and nine events. 
	Event labels represent the instantaneous time when that event occurred.
	b) The temporal components of (a) when $\mDt = 4$.
	c) The temporal barcode of (b). 
	There are three different components.
	Events in each component appear as black lines.
	Components 1 and 2 are distinct from 3 as they involve a distinct set of nodes.
	Components 1 and 2 are distinct as there is a gap greater than \Dt{} between activity on the nodes.}
	\label{fig:barcode}
	\end{figure}

	\subsection{Motif and Inter-event Time Distributions}

	We can also consider the IET and motif distributions across the components of the TEG\footnote{
		For consistency with the work of \cite{kovanen2011temporal} we will consider only \emph{valid} motifs and their corresponding IETs.
	}.
	
	The simple temporal networks in Figure~\ref{fig:examples} have trivial motif distributions.
	In Figure~\ref{fig:examples}(a) the only motif present is that of ABAC, reflective of the broadcasting type behaviour of node $\epsilon$ in this instant.
	If we were to consider the distribution of motifs in Figure~\ref{fig:examples}(b) we would see an equal split between the ABAB and ABBA motifs.
	However, considering the motif distribution of each component we see that there are in fact two distinct components containing either the ABAB or ABBA motif only.
	Without a suitable null model for the temporal network, analysing the motif distributions alone cannot give the significance of any observations \cite{shen2002network, milo2002network}, and choosing a null model is non-trivial beyond time-shuffling and time-reversal \cite{artzy2004comment, bajardi2011dynamical}.
	Comparing the temporal network with itself however allows us to gain information about the relative motif counts across components.
	Motif counts can be compared across different node or event types, or even different intervals in the network, however, given the use of temporal components in the calculation of motif counts, comparing the motif distributions across temporal components is a natural way to proceed.

	\begin{figure}[h!]
	\centering
	\begin{subfigure}{\linewidth}
	\includegraphics[width=\linewidth]{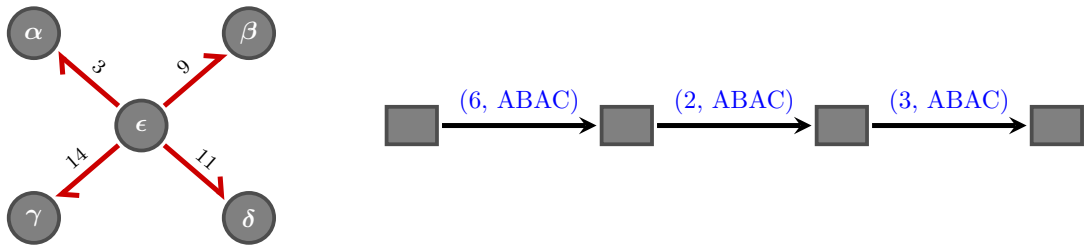}
	\caption{(Left) a temporal network consisting of a central node messaging four other nodes in turn. 
	(Right) the corresponding TEG.}
	\end{subfigure}
	\begin{subfigure}{\linewidth}
	\includegraphics[width=\linewidth]{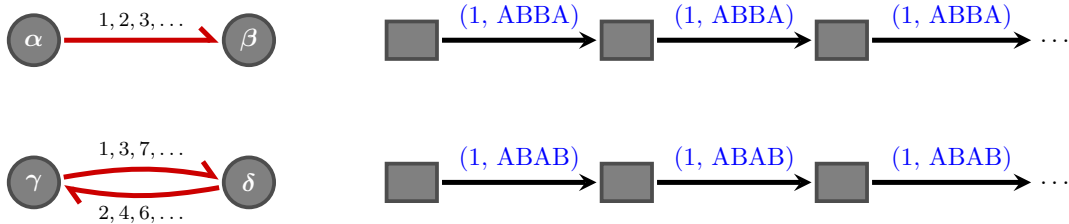}
	\caption{(Left) a temporal network consisting of two pairs of nodes. The bottom pair periodically reciprocate messages in turn, whereas in the top pair all messages are sent in one direction. 
	(Right) the corresponding TEG.}
	\end{subfigure}
	\caption{Examples of temporal networks and their temporal event graphs.}
	\label{fig:examples}
	\end{figure}

	Returning to the random temporal network example of Section~\ref{subsec:component_sizes} it can be shown that the motif distribution is given by 
	\begin{align}
	\Pr(x) = \begin{cases}
	 \frac{1}{4N-6} &\mbox{ for } x \in \{{\rm ABAB, ABBA}\} \\
	 \frac{N-2}{4N-6} &\mbox{ for } x \in \{{\rm ABBC, ABCB, ABAC, ABCA}\}.
	\end{cases}
	\label{eqn:motif_probabilities}
	\end{align}
	So, as $N \to \infty$, the ABAB and ABBA motifs are less likely to be observed and all other motifs are observed with equal probability.
	This illustrates why the random temporal network model is an unsuitable null model for social systems where one expects a degree of reciprocity. 

	Coupled to each motif, each edge in the TEG carries the IET between the two connected events.
	This is the time between events which an individual node participates (inward and outward activity).
	This time differs from what has been previously studied in temporal networks - usually the global time between events for the entire network, or the outward activity of an individual node \cite{jo2012circadian,barabasi2005origin,holme2013temporal}.
	We may also partition the IETs based on the motif formed between the two events.
	For example we can calculate $\Pr(t | m)$ with $m \in \mathbb{M}$, the probability of observing an IET of $t$ given the motif formed between the two events is $m$.
	By considering these conditional probabilities we can uncover more information about the process that generated the temporal network then if we considered the IETs and motifs in isolation.
	
	In Figure~\ref{fig:iet_distribution}(a) we plot the CCDF of the IETs of the TEG.
	For real data, this distribution is a complex function of node interactivity and activity patterns. 
	For the random temporal network however the distribution is approximately geometric.
	This is due to each node having a constant probability of being in an event at each iteration.
	The distribution would be exactly geometric if $X$ was deterministic.
	In Figure~\ref{fig:iet_distribution}(b) we see that the motifs with two nodes (ABAB and ABBA) occur faster on average than the motifs containing three nodes.
	This make sense as in the random temporal network model the three node motifs are more likely to occur and so the two node motifs must occur quickly or not at all.  
	
	\begin{figure}[h!]
	\centering
	\begin{subfigure}{0.49\linewidth}
	\centering
	\includegraphics[width=\linewidth]{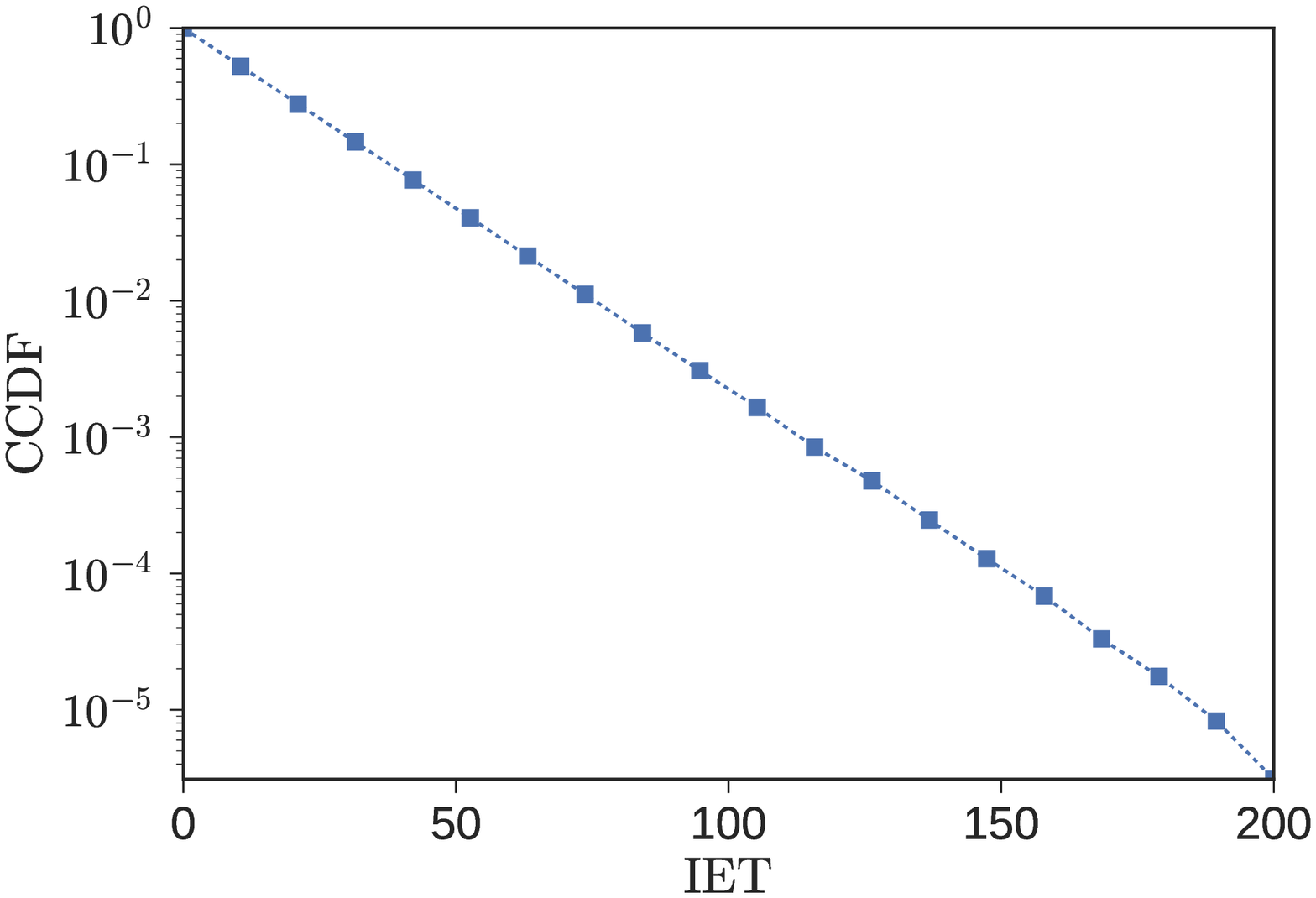}
	\caption{}
	\end{subfigure}
	\hfill
	\begin{subfigure}{0.49\linewidth}
	\centering
	\includegraphics[width=\linewidth]{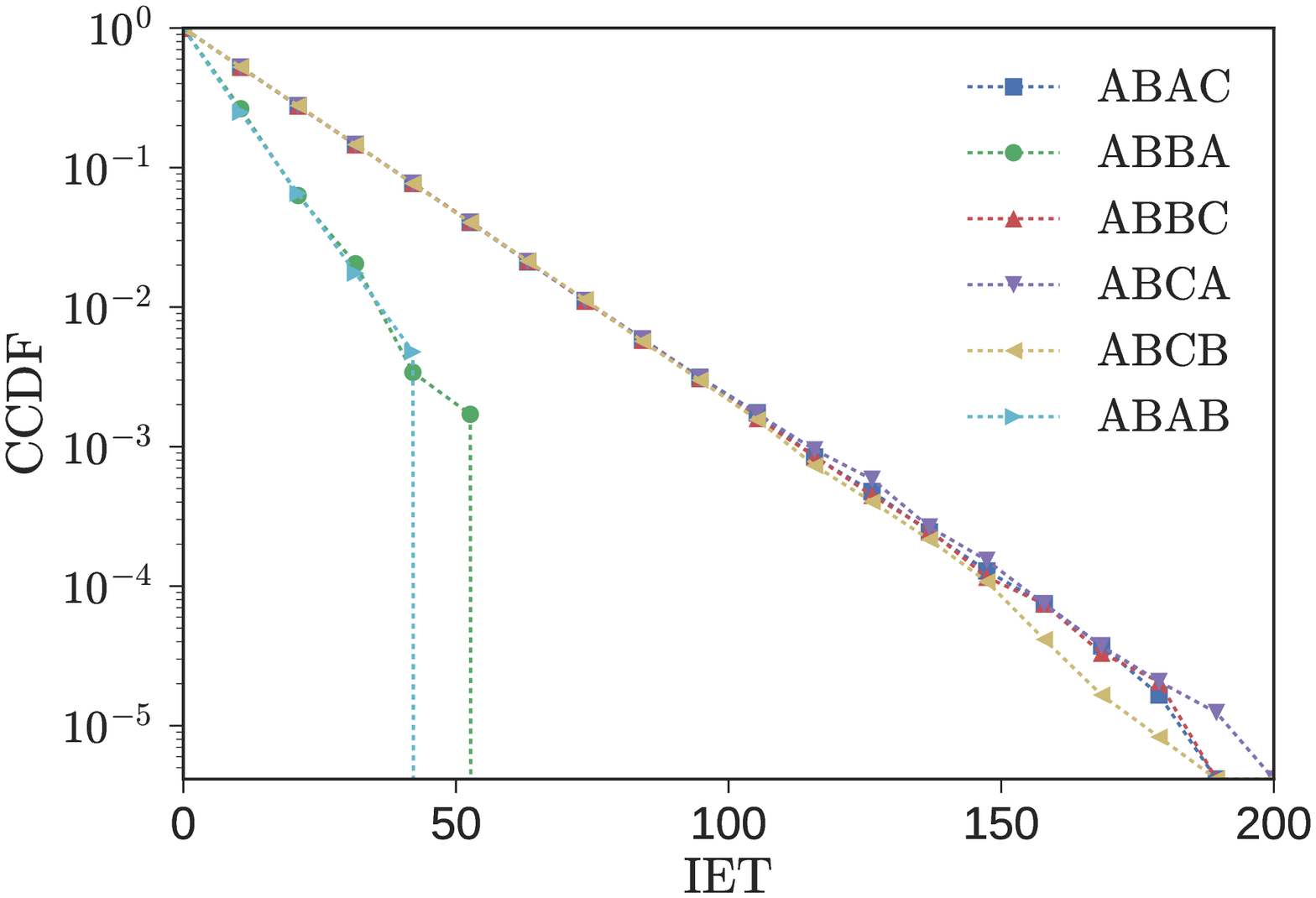}
	\caption{}
	\end{subfigure}
	\caption{The IET distributions for the random temporal network. 
	a) the CCDF for the IET distribution of the TEG, i.e., the time between consecutive events for each node.
	b) the CCDFs of the IET distributions, conditional on the motif formed.
	The motifs containing two nodes have on average a smaller IET than motifs with three nodes.}
	\label{fig:iet_distribution}
	\end{figure}

	\subsubsection{Entropic Measures}

	There are six possible two-event motifs. 
	However with generalisations of temporal networks such as allowing coloured edges or nodes there can be many more possible motifs.
	The full distribution of motifs can therefore be difficult to analyse and it is instead useful to use a single measure to capture the diversity (or predictability) of the motifs within a component.
	For this purpose we use Shannon's information entropy, which for a distribution $\{p\}$ is given by
	\begin{align}
	S\left(\{p\}\right) = - \sum_{p_i \in \{p\}} p_i \log_2 p_i.
	\end{align}
	This takes a minimal value of zero when $p_k=1$ and $p_i=0$ for all $i \neq k$, corresponding to a fully predictable system, and a maximal value of $-\log_2 p$ when all $p_i=p$, corresponding to uniform randomness.
	When there are six possible motifs the maximum entropy is $-\log_2 \left(\frac{1}{6}\right) \approx 2.58$.
	In both the examples in Figure~\ref{fig:examples} the entropy of each component is identically zero as each component consists of a single motif.
	In this sense, all these components are completely predictable.
	For the random temporal network there are four possible motifs occurring with equal probability in the large $N$ limit.
	The entropy is therefore $-\log_2\left(\frac{1}{4}\right) = 2$.
	Therefore, from a motif based viewpoint, the random temporal network is not as random as possible.

	Likewise, although less simply, we can compute the entropy of the IET distribution.
	As the IET is a continuous variable we instead use the cumulative residual entropy (CRE) \cite{drissi2008generalized,rao2004cumulative} defined as
	\begin{align}
	S_{\rm CRE}\left(X\right) = - \int \Pr(X>x) \log_2 \left[ \Pr(X>x) \right] \mathop{dx},
	\end{align}
	where $X$ is the IET distribution.
	The CRE shares many features with Shannon entropy (the CRE of a delta function is $0$, for example), and provides a consistent means to characterise the diversity in the IET distribution\footnote{
		The variance can also be considered, however this performs poorly as a measure of diversity on distributions not well described by their mean, e.g. a bimodal distribution. 
	}.

	\subsection{Induced Aggregate Networks}

	The \Dt{}-TEG provides a convenient way to decompose a larger temporal network, however being event-centric it can be difficult to assess the connectivity of the nodes within each component.
	This information can be extracted easily however by considering the static aggregation of the temporal component.
	The static network can then be analysed using standard methods to find quantities of interest.
	In particular, we will be interested in the number of nodes, edge density, the fraction of reciprocated edges, and network diameter.

	Studying the components of the decomposed network offers the advantage of understanding the role of nodes within a particular context, as opposed to consideration of the static graph of the full temporal network, which may be dense or noisy, or of fixed intervals which may dissect patterns of behaviour.
	Partitioning the random temporal network into intervals of fixed width results in a series of Erd\H{o}s-R\'enyi (ER) static networks with edge forming parameter $p$ dependent on the number of events in each partition.  
	This gives the `temporal ER network' as described in \cite{scellato2013evaluating}.
	The aggregated networks of the TEG components by contrast are not in the class of ER graphs and their properties are yet to be determined.

\section{Application to Data}
	\label{sec:data}

	In this section we consider the social network of students from University of California, Irvine (UCI) \cite{panzarasa2009patterns,opsahl2009clustering}\footnote{
		Data available at: \url{http://snap.stanford.edu/data/CollegeMsg.html}
	}.
	The social network was created to sustain social interaction among students and to help enlarge their social circles.
	Students created a profile which contained a short biography and demographic characteristics.
	Students could then view or message any other student in the network.

	The dataset covers a period of six months from April to October 2004, over this time $59,835$ messages were sent between $1,899$ users\footnote{
		Special users who broadcast messages to the entire network were removed.	
	}.
	The resulting aggregate network has $20,296$ directed edges, meaning the network is sparse ($0.56\%$ of all possible edges are present).

	To get a first impression of the network structure we look at the temporal barcode of the $10$min-TEG on over a short period of the data (in this case twelve hours).
	The TEG consists of multiple large components which occur over a duration of over an hour.
	Some of these components overlap in time suggesting that distinct conversations were occurring.
	Over the same time period there a number of smaller components of interest.
	For example, component $20$ consists of two users exchanging messages back and fourth over a period of $30$ minutes.
	Interestingly one of these users has a response time significantly shorter than the other.
	Component $19$ is a more complicated mix of broadcasting nodes which then converse amongst themselves.
	Despite containing more events than component $20$, this component only lasts for a duration of $9$ minutes.
	We can quantify these observations by examining the motif and IET distributions across each component.

	\begin{figure}[htb]
	\centering
	\includegraphics[width=\linewidth]{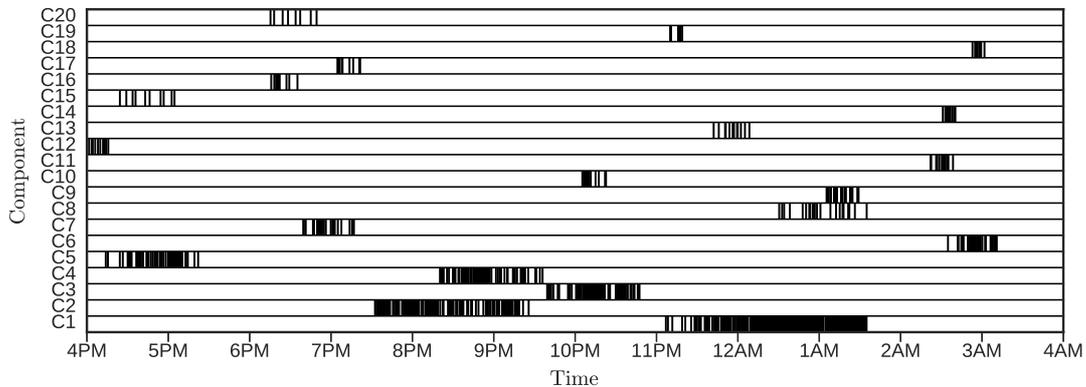}
	\caption{
	The top $20$ components of the $10$min-TEG for a twelve hour period beginning on $10$th May $2004$. 
	As with Figure~\ref{fig:barcode}, each vertical black line represents a single event.
	The $10$min-TEG consists of a number of large components (with some overlap) which occur in the mid to late evening.
	There are also many smaller components occurring at the same time with distinctive IETs.
	}
	\label{fig:data_barcode}
	\end{figure}

	We now systematically study the structural dependence of the TEG on the parameter \Dt{}.
	In Figure~\ref{fig:data_components}(a) we plot the size of the largest component (as a fraction of all events).
	As we increase \Dt{}, the largest component increases in a step-wise fashion, indicating that there are distinct timescales in the connectivity of the TEG.
	Jumps in connectivity around $\mDt = 6$ could be attributed to messages which occur just before night and then first thing in the morning.
	We also see that as $\mDt \to \infty$ the largest component contains all but four events (out $59,835$).
	The aggregate static network contains only four components with the activity of the smaller three consisting of only one or two events.
	\begin{figure}[htb]
	\centering
	\begin{subfigure}{0.49\linewidth}
	\centering
	\includegraphics[width=\linewidth]{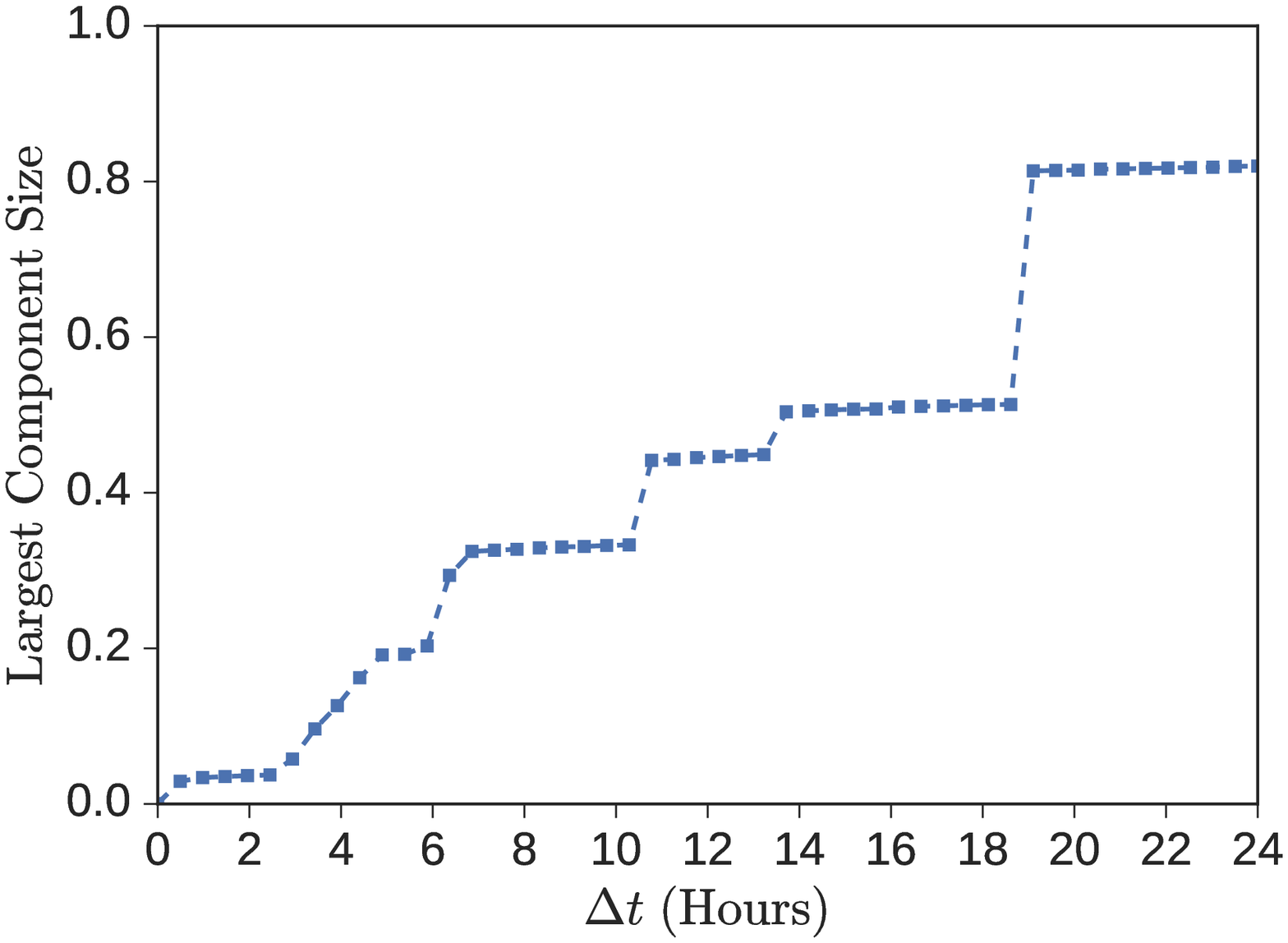}
	\caption{}
	\end{subfigure}
	\hfill
	\begin{subfigure}{0.49\linewidth}
	\centering
	\includegraphics[width=\linewidth]{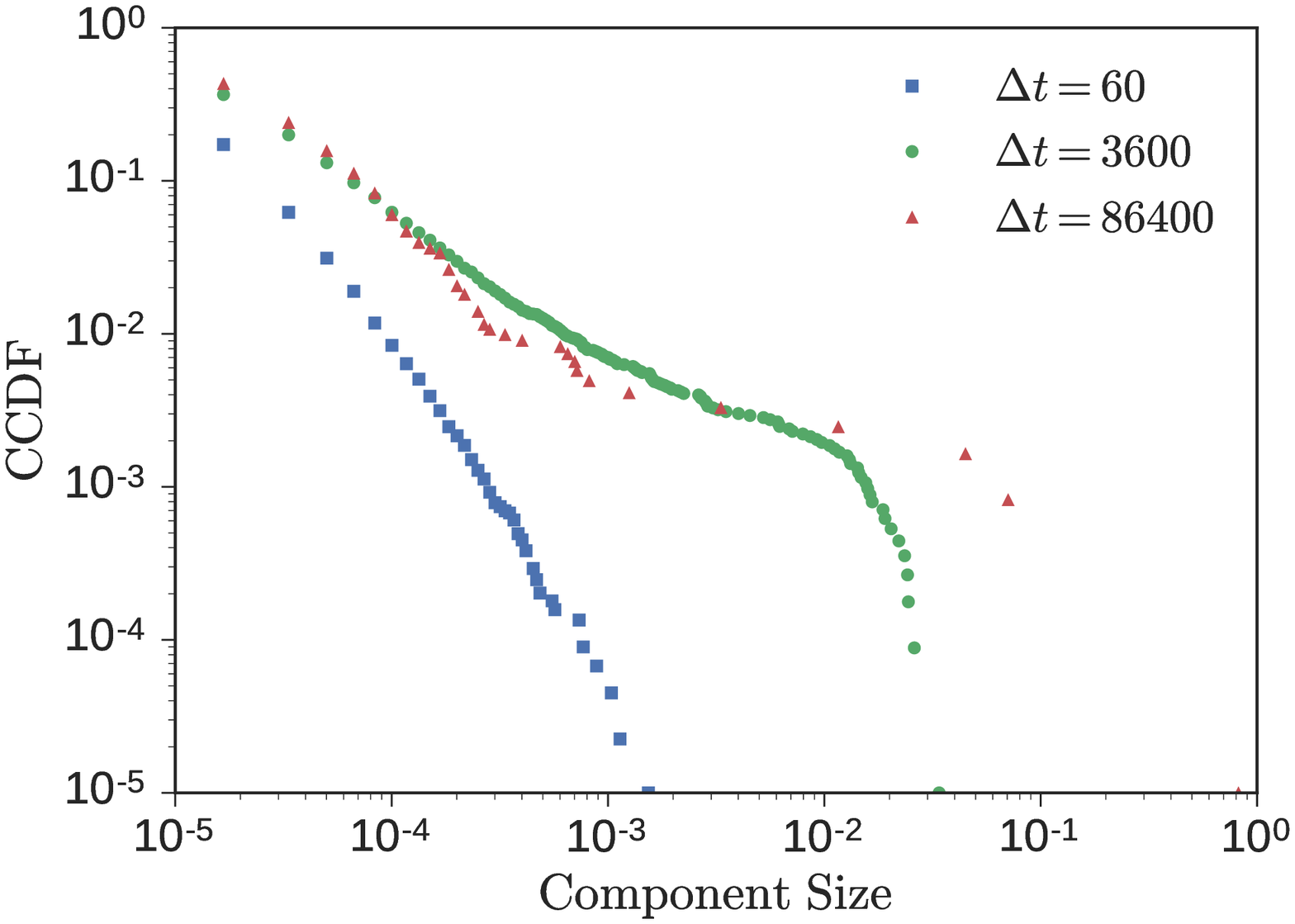}
	\caption{}
	\end{subfigure}
	\caption{Temporal component dependence on $\Delta t$. 
	(a) The size of the largest temporal component in the \Dt{}-TEG as a fraction of all events for the UCI network.
	(b) The corresponding distribution of temporal component sizes for $\mDt = 60,3600,86400$ seconds (corresponding to $1$ minute, $1$ hour and $1$ day). 
	For $\mDt=60$ the distribution of component sizes are well represented by a power-law distribution.
	For $\mDt=3600$ we see the onset of the giant component.
	For $\mDt=86400$ the largest component is over $80\%$ of all events and there are very few moderately sized components.
	The remaining components are all relatively small.
	}
	\label{fig:data_components}
	\end{figure}
	Furthermore the distribution of component sizes for varying \Dt{} is given in Figure~\ref{fig:data_components}(b).
	The component size distributions mirror that of the random temporal network with the formation of a giant component as \Dt{} increases.
	However in this case the transition from multiple small components to a giant component is less abrupt.
	For $\mDt=60$ (blue squares) and for other small values of \Dt{} the distribution of component sizes can be accurately modelled by power-law distribution.

	For the remainder of this section we fix $\mDt = 3600s$ ($1$ hour).
	The choice of \Dt{} in this case (as in previous work) is chosen arbitrarily, although as Figure~\ref{fig:data_components} confirms we are in a regime where the largest component is no larger than $5\%$ of the total number of events.
	We first calculate the distribution of two-event motifs across the entire network and compare this to the average distribution over an ensemble of $200$ time-shuffled versions of the network\footnote{
		Node pairs are kept the same however times are shuffled between events.
	} (Table~\ref{tab:motif_dist}).
	The distribution of motifs in the true network is differs significantly from the random ensemble (with a z-score of $46$).

	\begin{table}[htb]
	\centering
	\setlength\tabcolsep{10pt}
	\begin{tabular}{ccccccc} \toprule[1pt]
	\textbf{Network} & ABAB & ABBA & ABAC & ABCA & ABBC & ABCB \\ \midrule[0.5pt]
	{UCI} & $7.0 \times 10^{-2}$ & $0.14$ & $0.27$ & $0.15$ & $0.11$ & $0.25$ \\
	{Shuffled} & $9.0 \times 10^{-3}$ & $7.6 \times 10^{-3}$ & $0.27$ & $0.22$ & $0.22$ & $0.26$ \\ \bottomrule[0.5pt]
	\end{tabular}
	\caption{Motif distribution for the UCI temporal network.}
	\label{tab:motif_dist}
	\end{table}

	\begin{figure}[h!]
	\centering
	\includegraphics[width=0.7\linewidth]{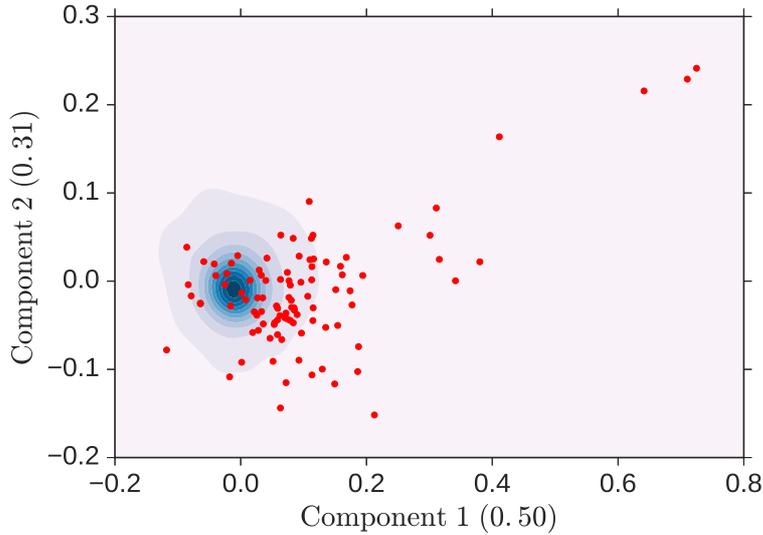}
	\caption{
		The motif distribution of the largest $100$ components of the $1$h-TEG, reduced to two dimensions using principal component analysis (red dots).
		Behind, a kernel density estimate of the motif distribution from the largest $100$ components of $200$ time-shuffled networks (blue shading).
		Darker areas have higher probability.
		Here we can see that the average motif distribution differs from the randomised networks, and that also there is a larger variance among the components compared to the randomised networks (see text).
	}
	\label{fig:data_motif_distribution}
	\end{figure}

	By considering the TEG structure (the temporal components in particular) we can decompose the motif counts to the components from which they originate (see Figure~\ref{fig:data_motif_distribution}).
	There are two observations to note.
	Firstly, the majority of the TEG components (red dots) lie outside the bulk of the time-shuffled component distribution (blue shading, darker being higher density) which confirms our earlier observation that the difference in average is significant.
	Secondly, the diversity in the motif counts of each component is greater than in the randomised networks (a z-score of $12$ when considering the average nearest-neighbour distance).
	The largest components have distributions close to the average for the entire network however there are certain components where one motif is more greatly expressed than the others.
	For example, the three components in the top right of Figure~\ref{fig:data_motif_distribution} the ABAC motif makes up over $95\%$ of all observed motifs in these components.
	This highlights that the temporal network is not simply made up of homogeneous groups of nodes and activity but instead consists of distinct heterogeneous components.
	The largest components can be decomposed by further reducing \Dt{} which may isolate more diverse component structures.

	Another benefit of studying IETs and motifs in tandem is that we can consider the IET distribution conditioned on the motif formed between the two events.
	In Figure~\ref{fig:data_iet_distribution}(a) we see the IET distribution across the whole TEG (blue $\square$), compared to the same distribution over an ensemble of $100$ time-shuffled versions of the network (dashed).
	The time-shuffled CCDF is well modelled by a log-normal distribution, however in the real network smaller IETs are overrepresented and the log-normal fit is poor.
	\begin{figure}[h!]
	\centering
	\begin{subfigure}{0.49\linewidth}
	\centering
	\includegraphics[width=\linewidth]{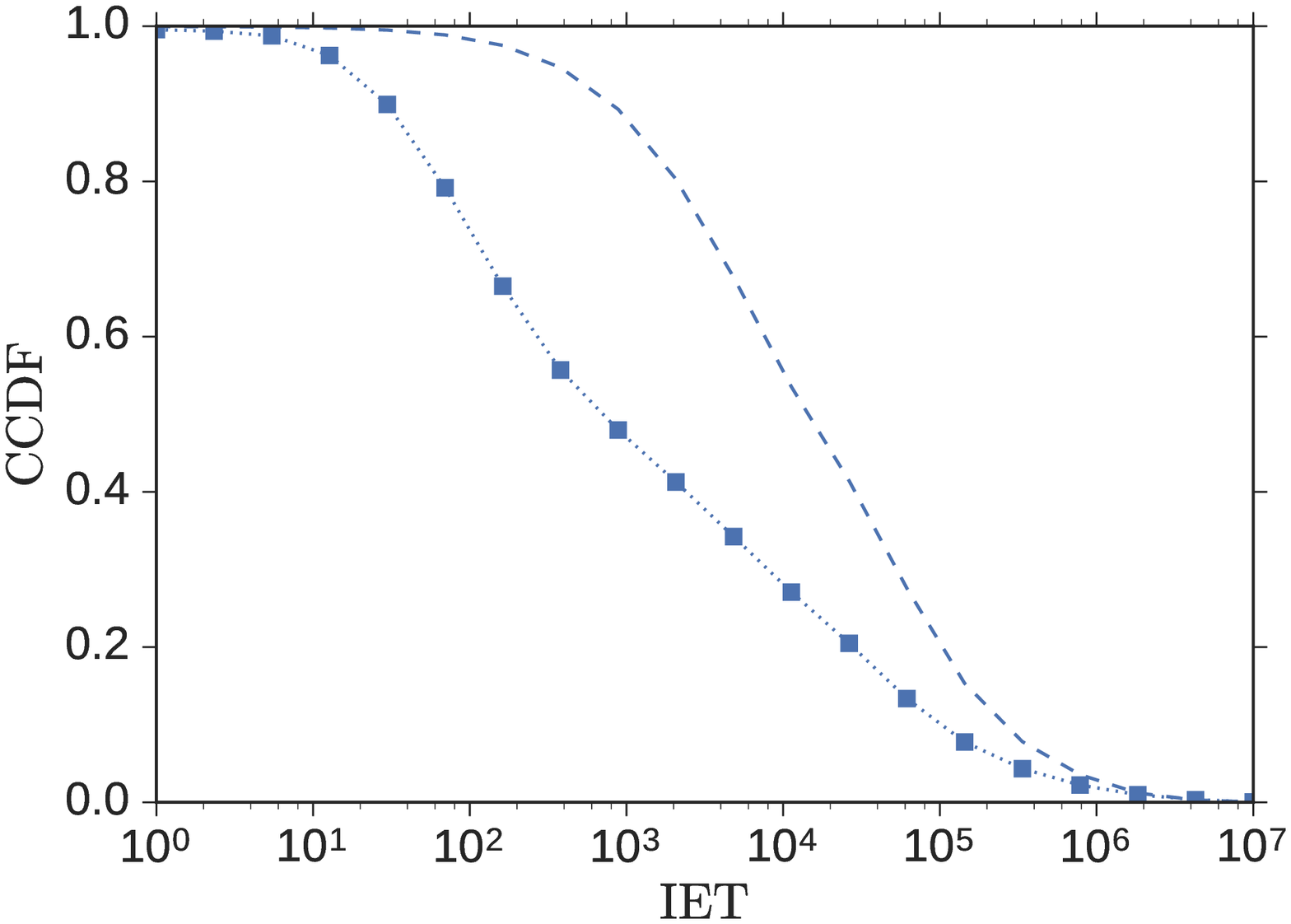}
	\caption{}
	\end{subfigure}
	\hfill
	\begin{subfigure}{0.49\linewidth}
	\centering
	\includegraphics[width=\linewidth]{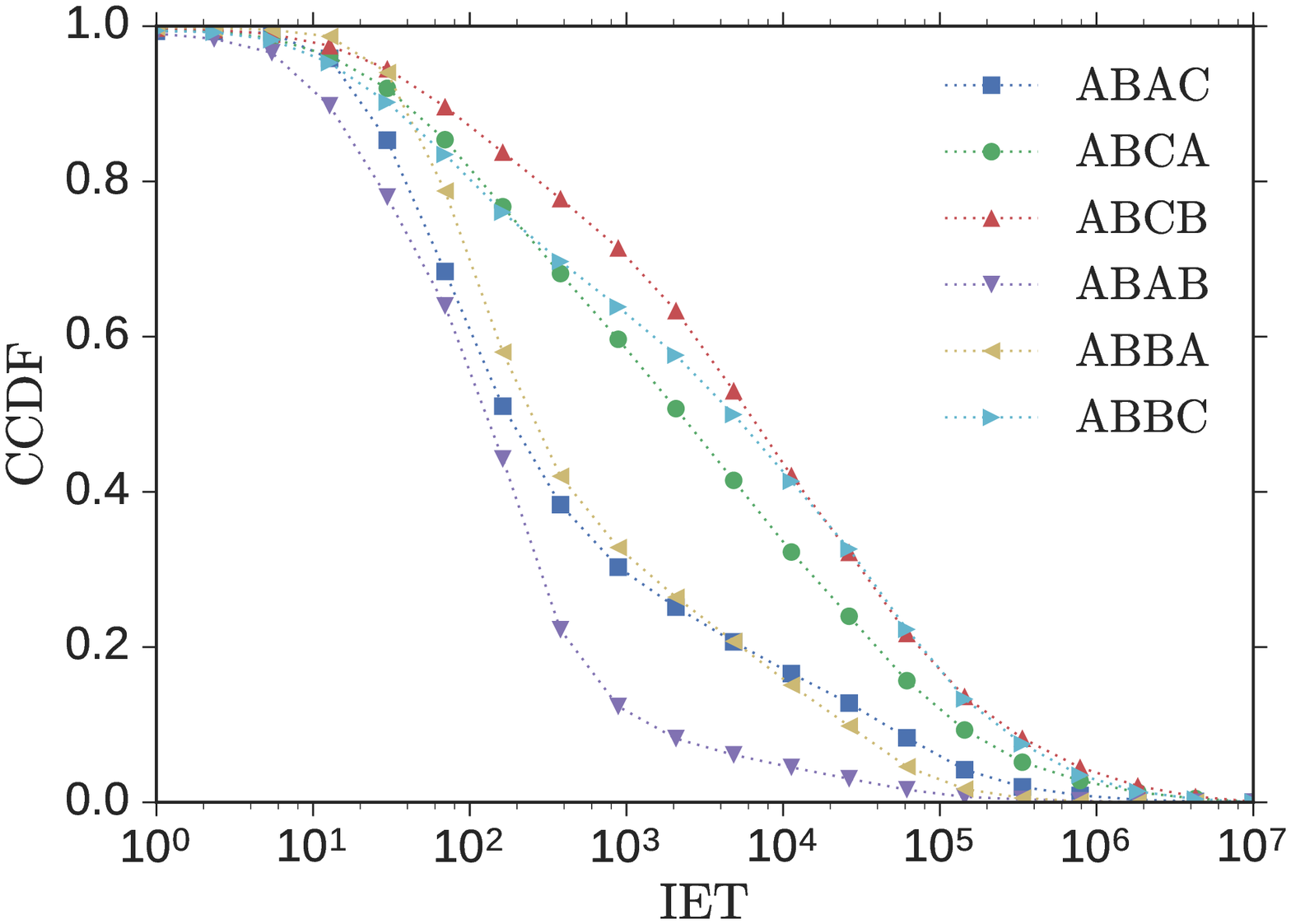}
	\caption{}
	\end{subfigure}
	\caption{The IET distributions for the UCI network. 
	(Left) the CCDF for the IET distribution of the TEG (blue $\square$), i.e., the time between consecutive events for each node.
	The IET CCDF of an ensemble of $100$ time-shuffled versions of the temporal network is given by the dashed line.
	(Right) the CCDFs of the IET distributions, conditional on the motif formed.
	The fastest appearing motif on average is that of ABAB.
	The slowest appearing motif is the ABCB motif.
	The time-shuffled ensemble CCDFs for each motif are all roughly identical (not shown).  
	}
	\label{fig:data_iet_distribution}
	\end{figure}
	By considering the conditional IET distributions, $\Pr(t|m)$, we see in Figure~\ref{fig:data_iet_distribution}(b) that on average the ABAB motif occurs on average much quicker than the other five motifs, and all motif IETs are smaller than the random ensemble equivalent.
	This is most likely due to users of the network breaking their messages and sending the same information over multiple messages.

\section{Conclusions}
	\label{sec:conclusions}

	In this article we introduced the temporal event graph, a static representation of a temporal network.
	Furthermore we show that the TEG can uniquely define a temporal network.
	In this sense we are able to fully describe a temporal network purely in terms of the inter-event times and two-event motifs.
	In Section~\ref{sec:data} we showed that the TEG provided a natural decomposition of the temporal network and that by considering the inter-event times conditioned on motif type we were able to uncover different timescales for behaviour that would not be visible when considering both properties independently.
	
	It is also worth noting that the TEG is not limited to simple event tuples $(u,v,t)$ but can be generalised in the same fashion as temporal motifs to include coloured events or nodes, e.g. to distinguish phone calls and SMS messages in communication networks.
	Provided a meaningful relationship between events exists then in fact any such sequence of time-stamped events can be represented by a TEG.
	The calculation of the TEG is also computationally efficient.
	Building the TEG from a temporal network can be done in time which scales linearly with the number of events in the network.
	This means that this type of analysis is well suited to large datasets such as those extracted from social or telecommunication networks.
	It also allows the TEG to be constructed in real-time (provided the data is received sequentially) and so provides a method to quickly assess behavioural changes in the network.

	While many details of the TEG are yet to be explored, our study demonstrates how a temporal network can be represented exactly as a static network and be classified using event relationships.
	This finding provides an initial, but significant step towards the systematic investigation of temporal networks and their generating mechanisms.

\section*{Funding}

This work was supported by an Engineering and Physical Sciences Research Council CASE Studentship [Grant number EP/L50550X/1].
Funding from Bloom Agency\footnote{
	\href{http://www.bloomagency.co.uk/}{bloomagency.co.uk}
} 
also kindly acknowledged.

\section*{Acknowledgements}

We thank Jonathan Ward for the careful review of this manuscript and Alastair Rucklidge, Mauro Mobilia, Mary Aprahamian, and Charles Taylor for useful discussion and insightful offerings.

\end{document}